\documentclass[11pt]{article}

\usepackage[margin=1in]{geometry}
\usepackage{amsmath}
\usepackage{amssymb}
\usepackage{amsthm}
\usepackage{natbib}
\usepackage{graphicx}
\usepackage[bookmarksnumbered,colorlinks,bookmarks,citecolor=blue,linkcolor=blue,urlcolor=blue,breaklinks,linktocpage]{hyperref}
\usepackage{booktabs}
\usepackage{algorithm}
\usepackage{algpseudocode}
\usepackage{float}

\graphicspath{{figures/}{docs/figures/}{./}}

\title{Tonal parsimony in chord-sequence analysis:\\
combining modulation cost and tonal vocabulary}

\author{Fran\c{c}ois Pachet\thanks{Corresponding author. Email: pachet@gmail.com}\\
LIP6, Sorbonne Universit\'e, Paris, France\\
Ynosound, Paris, France}
\date{}

\newcommand{\cost}{C}
\newcommand{\nval}{K}
\newcommand{\cand}{T}
\newcommand{\keys}{\mathcal{T}}

\theoremstyle{plain}
\newtheorem{proposition}{Proposition}[section]

\newenvironment{keywords}{\paragraph*{Keywords:}}{}

\begin{document}
\maketitle

\begin{abstract}
We study the assignment of local tonalities to chord sequences, a task useful
for harmonic analysis, composition, and jazz-oriented improvisation. Standard
dynamic-programming approaches minimize modulations but can introduce
unnecessarily many tonal centers. We compare this transition-only objective
with pure minimum-vocabulary analysis and with tonal parsimony, which minimizes
lexicographically the number of modulations and then the number of distinct
tonalities. Although this joint objective is combinatorially hard in general,
we give exact algorithms exploiting the fixed 24-tonality major/minor universe.
On 31,032 LMD Chords sequences, tonal parsimony preserves the transition
optimum while reducing tonal vocabulary in 55.8\% of cases. With weighted
jazz-substitution closure, it lowers mean tonalities from 3.802 to 3.206 and
modulations from 16.728 to 12.141. On 1,555 annotated jazz standards, it
improves compatible chord-scale agreement to 95.6\%, supporting tractable
professional-scale harmonic analysis.
\end{abstract}


\section{Introduction}

Given a chord sequence, a harmonic analysis often aims to explain each chord as
belonging to one or more local tonal regions.  For example, the chord
\texttt{Am7} may be compatible with C major, A minor, G major, and other
tonalities depending on the chord vocabulary used.  A dominant seventh chord
such as \texttt{D7} naturally points toward G, but may also be interpreted in
other contexts.  In jazz and related idioms, such an analysis has an immediate
practical use: the assigned tonality can be mapped to a scale, providing a
compact guide for improvisation, reharmonization, or generative composition.

This paper separates the local compatibility layer from the global optimization
layer.  For each chord position $i$, we assume that a separate process has
already computed a set $\cand_i$ of candidate tonalities.  In the baseline
experiments reported below, these candidates are drawn from the usual 12 major
and 12 minor tonalities using a direct chord-scale compatibility relation.  We
then study an extension in which this relation is enriched by weighted jazz
substitution rules.  In both cases, the task studied here is the global one:
choose one tonality $x_i \in \cand_i$ for each chord while optimizing the
sequence-level criteria.

Dynamic-programming and optimal-segmentation approaches to local key or
tonality analysis have a long history.  They commonly combine local evidence
with a transition penalty, or equivalently prefer analyses with fewer key
changes.  Temperley's reconsideration of the Krumhansl--Schmuckler key-finding
algorithm explicitly introduces a modulation penalty between successive
segments \citep{temperley1999}.  The jazz harmonic analysis of
\citet{choi2011} is close to the present setting: harmonic analysis is
formulated as an optimal tonality-segmentation problem over a chord sequence.
Recent
work on joint pitch spelling and local/global key estimation also uses
dynamic-programming style optimization over musical sequences
\citep{bouquillard2024}.

Probabilistic approaches provide a related inference perspective.
\citet{raphaelstoddard2004} formulate functional harmonic analysis as
inference in a probabilistic model over keys and harmonic functions.  More
generally, tonal-space representations such as Chew's Spiral Array provide
formal models of tonal centers and tonal distance \citep{chew2014}, while
statistical chord-sequence models study latent harmonic categories and
self-emergent grammar \citep{tsushima2018}.  Grammar-based accounts of jazz
harmony, notably the generative grammar for jazz chord sequences of
\citet{steedman1984}, provide a complementary formal view of substitutional and
recursive harmonic structure.

Existing dynamic-programming approaches to harmonic analysis penalize
modulations or segment transitions, but do not optimize the cardinality of the
set of tonal centers used in the analysis.  We introduce a tonal-parsimony
criterion formulated as a lexicographic optimization of modulation cost and
tonal vocabulary size.  This captures the musical intuition that key has
inertia, while adding a second intuition: among equally smooth analyses, one
should prefer a smaller tonal vocabulary.

The paper makes three contributions.  First, it formulates tonal parsimony as
the lexicographic objective
\[
  \min_{\mathrm{lex}}(\cost,\nval),
\]
where $\cost$ is the modulation cost and $\nval$ is the number of distinct
tonalities used.  Second, it gives a specialized exact dynamic program for this
objective, realizing the combination of a cost-regular path criterion and an
\textsc{NValue} criterion over a fixed tonality universe.  Third, it introduces
a substitution-closure layer that expands local candidate domains with weighted
jazz-harmonic reinterpretations before global optimization.

The strongest empirical test is the external chord-scale validation: on 1,555
professionally annotated jazz standards~\citep{jazzstandardsbook}, tonal
parsimony reaches 81.1\% strict and 95.6\% conservative compatible agreement,
outperforming both baselines while preserving the minimum modulation cost of
the transition-only dynamic program.  This indicates that the criterion is not
only a more compact objective, but also a better predictor of expert local
scale annotations.

\section{Problem Statement}

Let
\[
  c_1,\ldots,c_n
\]
be a chord sequence.  For each chord $c_i$, let
\[
  \cand_i \subseteq \keys
\]
be its set of candidate tonalities, where in the present work
$|\keys|=24$.  An analysis is a sequence
\[
  x_1,\ldots,x_n
\]
such that $x_i \in \cand_i$ for all $i$.

We define two quantities:
\[
  \cost(x) = \sum_{i=1}^{n-1} [x_i \ne x_{i+1}],
\]
the number of modulations or adjacent tonality changes, and
\[
  \nval(x) = |\{x_1,\ldots,x_n\}|,
\]
the number of distinct tonalities used by the analysis.

The classical dynamic-programming objective minimizes $\cost$.  Our proposed
objective minimizes $(\cost,\nval)$ lexicographically: first minimize the number
of modulations, and among all equally smooth analyses choose one using fewer
tonalities.

\section{Constraint Satisfaction Formulation}

The assignment problem can be modeled as a finite-domain constraint
optimization problem \citep{rossi2006}.  There is one decision variable $x_i$
per chord position, with finite domain $\cand_i \subseteq \keys$.  The hard
constraints are the domain constraints $x_i \in \cand_i$; the musical
criteria are expressed as global quantities over the whole sequence.

The modulation cost $\cost(x)$ is a regular sequence cost: it can be represented
by a small automaton that reads the selected tonalities and accumulates a unit
cost whenever two consecutive labels differ.  This places the transition-only
objective in the family of \textsc{Regular} and cost-regular global constraints
\citep{pesant2004,demassey2006}.  The number of distinct tonalities
$\nval(x)$ is the value counted by the \textsc{NValue} global constraint
\citep{bessiere2006}.  This follows the broader constraint-programming practice
of treating global constraints as structured relations whose graph or automaton
structure can be exploited algorithmically \citep{beldiceanu2000}.

This constraint-programming view is primarily a modeling statement.  A literal
off-the-shelf CP model combining a cost-regular path constraint with an
\textsc{NValue} objective is unlikely to be tractable at the scale considered
here: the solver would still have to coordinate a sequence cost with a global
cardinality objective over thousands of instances.  The algorithms below are
specialized exact algorithms inspired by this model.  They exploit the fixed
24-tonality universe directly, instead of relying on generic
search over the combined global constraints.

As a sanity check, we wrote the direct model in MiniZinc using
\textsc{costRegular} and \textsc{NValue}.  It reproduces the expected optima on
small examples, including the substitution-enriched progression
\texttt{Cmaj7 F\#7 F Bb7 Em Am Dm Db7}, for which it returns a single C-major
path.  Table~\ref{tab:minizinc-sanity} gives exploratory timings with
MiniZinc/Gecode 6.3.0 and the specialized implementation.

\begin{table}[h]
\centering
\scriptsize
\begin{tabular}{lrrrr}
\toprule
Instance & Chords & Tonalities & MiniZinc/Gecode & Specialized Python \\
\midrule
Tie-break example, $(\cost,\nval)=(2,2)$ & 3 & 3 & 0.18 s & 0.013 ms \\
Substitution example, $(\cost,\nval)=(0,1)$ & 8 & 4 & 0.16 s & 0.019 ms \\
Corpus excerpt, $(\cost,\nval)=(17,5)$ & 32 & 22 & 21.9 s & 0.26 ms \\
Corpus excerpt, $(\cost,\nval)=(25,8)$ & 48 & 22 & $>2$ h 15 min & 0.61 ms \\
\bottomrule
\end{tabular}
\caption{Direct MiniZinc/Gecode model versus our specialized fixed-universe
implementation written in Python, measured on a MacBook Pro with an M1 Pro
processor.  The 48-chord MiniZinc run had not returned a solution after two
hours and fifteen minutes.}
\label{tab:minizinc-sanity}
\end{table}

This confirms that the CP formulation provides a declarative specification and
validation baseline, but that scalability requires the fixed-universe algorithm
used in the experiments below.

The three methods studied below correspond to three objective choices over the
same constraint model:
\[
  \min \cost(x), \qquad
  \min \nval(x), \qquad
  \min_{\mathrm{lex}}(\cost(x),\nval(x)).
\]
Thus the contribution is not a new compatibility model for chords and scales,
but a comparison of three constraint-optimization objectives for selecting one
tonality from each candidate domain.

\section{Three Methods}

\subsection{Transition-only dynamic programming}

The first method is the standard minimum-modulation path:
\[
  \min_x \cost(x).
\]
This can be solved by dynamic programming over positions and current
tonality.  Let
\[
  F_i(t)
\]
be the minimum modulation cost of a partial analysis ending at position $i$
with tonality $t$.  The Bellman recurrence is
\[
  F_1(t)=0 \qquad (t\in \cand_1),
\]
and for $i>1$,
\[
  F_i(t)=
  \min_{s\in \cand_{i-1}}
  \left(F_{i-1}(s) + [s\ne t]\right)
  \qquad (t\in \cand_i).
\]
The best final cost is $\min_{t\in\cand_n}F_n(t)$, and the path is recovered
by storing the predecessor $s$ that attains each minimum.  If $m=|\keys|$, the
time complexity is
\[
  O(nm^2)
\]
or less when the candidate domains are sparse.  This method is efficient and
often musically plausible, but it does not optimize $\nval$.  In constraint
programming terms, it is a shortest-path realization of a cost-regular
objective over the sequence of tonality variables.

\begin{algorithm}[H]
\caption{Transition-only dynamic programming}
\label{alg:transition-dp}
\begin{algorithmic}[1]
\Require Candidate domains $\cand_1,\ldots,\cand_n$
\Ensure A path $x_1,\ldots,x_n$ minimizing $\cost$
\ForAll{$t\in\cand_1$}
  \State $F_1(t)\gets 0$
  \State $P_1(t)\gets \bot$
\EndFor
\For{$i\gets 2$ to $n$}
  \ForAll{$t\in\cand_i$}
    \State $F_i(t)\gets \min_{s\in\cand_{i-1}}\left(F_{i-1}(s)+[s\ne t]\right)$
    \State $P_i(t)\gets \arg\min_{s\in\cand_{i-1}}\left(F_{i-1}(s)+[s\ne t]\right)$
  \EndFor
\EndFor
\State $x_n\gets \arg\min_{t\in\cand_n}F_n(t)$
\For{$i\gets n$ down to $2$}
  \State $x_{i-1}\gets P_i(x_i)$
\EndFor
\State \Return $x_1,\ldots,x_n$
\end{algorithmic}
\end{algorithm}

\subsection{Pure NValue minimization}

The second method ignores transitions entirely.  It searches for the smallest
set $S \subseteq \keys$ such that every chord has at least one candidate in
$S$:
\[
  S \cap \cand_i \ne \emptyset
  \qquad \text{for all } i.
\]
This is a hitting-set problem \citep{karp1972} and is equivalent to minimizing
the value of an \textsc{NValue} constraint \citep{bessiere2006} on the sequence
of tonality variables.  Once a
minimum-cardinality set $S$ is found, each chord is assigned to one compatible
tonality in $S \cap \cand_i$.

In the general case, hitting set is NP-hard \citep{karp1972}.  In the present
musical setting, however, $m=24$ is fixed.  Candidate sets are represented as
bit masks, and the search proceeds by increasing cardinality.  The dependence
on $m$ is exponential, but the dependence on sequence length has the
fixed-parameter form $O(f(m)n)$; a direct subset-search implementation is
$O(2^m n)$ up to small polynomial factors in $m$.

This method provides an informative extreme baseline.  It minimizes the tonal
vocabulary but may produce implausibly frequent modulations, since it ignores
adjacency altogether.

\begin{algorithm}[H]
\caption{Hitting-set \textsc{NValue} minimization}
\label{alg:hitting-set}
\begin{algorithmic}[1]
\Require Candidate domains $\cand_1,\ldots,\cand_n$ over universe $\keys$
\Ensure A path $x_1,\ldots,x_n$ minimizing $\nval$
\State Convert each domain $\cand_i$ into a bit mask $M_i$
\State $R\gets$ the inclusion-minimal masks among $\{M_1,\ldots,M_n\}$
\For{$k\gets 1$ to $|\keys|$}
  \State $S\gets \Call{Hit}{0,k,R}$
  \If{$S\ne\bot$}
    \For{$i\gets 1$ to $n$}
      \State $x_i\gets$ a deterministic tonality in $S\cap M_i$
    \EndFor
    \State \Return $x_1,\ldots,x_n$
  \EndIf
\EndFor
\State \Return failure
\Function{Hit}{$S,k,R$}
  \If{$S\cap M\ne\emptyset$ for all $M\in R$}
    \State \Return $S$
  \EndIf
  \If{$|S|=k$}
    \State \Return $\bot$
  \EndIf
  \State $M^\star\gets$ an uncovered mask in $R$ with minimum cardinality
  \ForAll{$t\in M^\star$}
    \State $S'\gets \Call{Hit}{S\cup\{t\},k,R}$
    \If{$S'\ne\bot$}
      \State \Return $S'$
    \EndIf
  \EndFor
  \State \Return $\bot$
\EndFunction
\end{algorithmic}
\end{algorithm}

The first feasible cardinality $k$ is the optimum value of $\nval$.  The final
assignment is not optimized for transitions; its path is only one witness that
the chosen tonal vocabulary covers the whole sequence.

\subsection{Tonal parsimony: minimum modulation then NValue}

The third method minimizes:
\[
  \min_{\mathrm{lex}}(\cost(x),\nval(x)).
\]
Equivalently, it asks: among all minimum-modulation analyses, which uses the
fewest distinct tonalities?

A generic formulation combines a path-cost constraint, as in cost-regular, with
an \textsc{NValue} constraint.  As noted above, this formulation gives a
declarative specification; the computation below uses a fixed-universe exact
algorithm.  The label formulation uses states of the form
\[
  (i,t,S,\cost),
\]
where $i$ is the position, $t$ is the current tonality, $S$ is the set of
tonalities already used, and $\cost$ is the accumulated modulation cost.  Since
there are only 24 tonalities, $S$ is represented as a bit mask.

Dominance pruning is crucial.  A label $A=(\cost_A,S_A)$ dominates another
label $B=(\cost_B,S_B)$ at the same position and current tonality if
\[
  \cost_A \le \cost_B
  \quad\text{and}\quad
  S_A \subseteq S_B.
\]
The dominated label can be discarded because no continuation can make it
better than the dominating one.

The label formulation is fixed-parameter tractable in the number $m$ of
tonalities: one can store, for each position and current tonality, the
nondominated used-tonality masks, giving $O(f(m)n)$ time, for example
$O(nm^2 2^m)$ in a straightforward implementation.  The dependence on $m$ is
exponential, but with fixed $m=24$ and small chord domains it is fast enough in
practice.  When $m$ is part of the input, the lexicographic problem is
NP-hard by reduction from Hitting Set~\citep{karp1972}: alternate each
hitting-set subset domain with a singleton separator tonality; every feasible
path has the same transition cost, and minimizing $\nval$ then chooses a
minimum hitting set plus the separator.  For the 0/1 transition cost, a path
can be viewed as a sequence of
constant-tonality segments: first minimize the number of segments, then apply
\textsc{NValue} pruning only among minimum-segment paths.  This keeps the
algorithm tied to the constraint model while avoiding unnecessary enumeration
of dominated paths.

The implemented version uses this segment view.  A segment is an interval of
consecutive chords that can all be assigned the same tonality.  Since
$\cost$ counts changes between adjacent positions, minimizing $\cost$ is
equivalent to minimizing the number of constant-tonality segments.

\begin{algorithm}[H]
\caption{Tonal parsimony by minimum segments and \textsc{NValue} pruning}
\label{alg:tonal-parsimony}
\begin{algorithmic}[1]
\Require Candidate domains $\cand_1,\ldots,\cand_n$
\Ensure A path minimizing $(\cost,\nval)$ lexicographically
\State Convert each domain $\cand_i$ into a bit mask $M_i$
\State $\mathcal{E}\gets\emptyset$
\For{$a\gets 1$ to $n$}
  \State $Q\gets$ all-tonalities mask
  \For{$b\gets a$ to $n$}
    \State $Q\gets Q\cap M_b$ \Comment{$Q=\bigcap_{j=a}^{b}M_j$}
    \If{$Q=\emptyset$}
      \State \textbf{break}
    \EndIf
    \State Add segment $(a,b,Q)$ to $\mathcal{E}$
  \EndFor
\EndFor
\State $B(n+1)\gets 0$
\For{$a\gets n$ down to $1$}
  \State $B(a)\gets 1+\min\{B(b+1):(a,b,Q)\in\mathcal{E}\}$
\EndFor
\State $L_1\gets \{(\emptyset,\bot)\}$ and $L_a\gets\emptyset$ for $a>1$
\For{$a\gets 1$ to $n$}
  \ForAll{$(S,\pi)\in L_a$}
    \ForAll{$(a,b,Q)\in\mathcal{E}$}
      \If{$\mathrm{segments}(\pi)+1+B(b+1)=B(1)$}
        \ForAll{$t\in Q$}
          \State $\pi'\gets$ append segment $(a,b,t)$ to $\pi$
          \State Add label $(S\cup\{t\},\pi')$ to $L_{b+1}$
          \State Remove from $L_{b+1}$ any label whose used-tonality set is a superset of another
        \EndFor
      \EndIf
    \EndFor
  \EndFor
\EndFor
\State Choose $(S^\star,\pi^\star)\in L_{n+1}$ minimizing $|S^\star|$
\State \Return path reconstructed from $\pi^\star$
\end{algorithmic}
\end{algorithm}

The masks $M_i$ are used only in the first phase, where they are accumulated
into segment masks $Q$.  A segment $(a,b,Q)$ records all tonalities that can
cover every chord from position $a$ to position $b$; the subsequent dynamic
program therefore no longer needs to inspect the original chord domains.

\begin{proposition}
Algorithm~\ref{alg:tonal-parsimony} returns an analysis minimizing
$(\cost,\nval)$ lexicographically.
\end{proposition}

\begin{proof}[Proof sketch]
Any valid path can be decomposed into maximal constant-tonality segments, and
any segment $(a,b,Q)$ generated by the first phase can be assigned any tonality
$t\in Q$ to obtain a valid constant-tonality interval.  Therefore minimizing
$\cost$ is equivalent to minimizing the number of such segments, minus one.
The backward dynamic program $B$ computes the minimum number of segments needed
from every position to the end.  The forward phase only extends partial
segmentations that can still reach the global optimum $B(1)$, so every complete
label in $L_{n+1}$ has minimum modulation cost, and every minimum-modulation
segmentation is represented.  The pruning step is safe because, at a fixed
boundary, a label whose used-tonality set is a superset of another can never
lead to a smaller final used set under any common continuation.  Choosing a
complete label with minimum $|S|$ therefore minimizes $\nval$ among all
minimum-$\cost$ paths.
\end{proof}

\section{Baseline Corpus Evaluation}

We first evaluated the three methods on the chord-symbol sequences from the
LMD Chords corpus~\citep{holloway2025lmdchords}.  This corpus contains chord
sequences extracted from selected Lakh MIDI Dataset files~\citep{raffel2016}
using the Chordino method~\citep{mauch2010}; in the present experiment we use
only the ordered chord labels, not timing or metadata.  Candidate tonalities
were held constant across the three methods, so the comparison isolates the
effect of the optimization objective.

\begin{center}
\begin{tabular}{lr}
\toprule
Quantity & Value \\
\midrule
Sequences read & 31,032 \\
Sequences analyzed & 31,017 \\
Chords analyzed & 2,499,035 \\
Unique chord names & 192 \\
Skipped empty sequences & 15 \\
Skipped sequences with no candidates & 0 \\
\bottomrule
\end{tabular}
\end{center}

\subsection{Aggregate results}

\begin{center}
\begin{tabular}{lrrrr}
\toprule
Method & Mean $\nval$ & Median $\nval$ & Mean $\cost$ & Median $\cost$ \\
\midrule
Transition-only DP & 4.829 & 4.000 & 16.728 & 12.000 \\
Hitting-set NValue & 3.123 & 3.000 & 33.274 & 26.000 \\
Tonal parsimony & 3.802 & 3.000 & 16.728 & 12.000 \\
\bottomrule
\end{tabular}
\end{center}

The hitting-set method obtains the smallest tonal vocabulary, as expected, but
nearly doubles the mean number of transitions.  Tonal parsimony preserves the
transition count of the former dynamic program while reducing the number of
distinct tonalities.

\subsection{Pairwise comparison against transition-only DP}

\begin{center}
\begin{tabular}{lrr}
\toprule
Statistic & Hitting-set NValue & Tonal parsimony \\
\midrule
Different paths & 85.1\% & 81.6\% \\
Fewer tonalities & 66.5\% & 55.8\% \\
Total $\nval$ reduction & 52,912 & 31,853 \\
Mean $\nval$ reduction & 1.706 & 1.027 \\
Median $\nval$ reduction & 1.000 & 1.000 \\
Max $\nval$ reduction & 13 & 9 \\
Same transition count & 26.7\% & 100.0\% \\
Higher transition count & 22,744 sequences & 0 sequences \\
\bottomrule
\end{tabular}
\end{center}

The key result is that tonal parsimony weakly dominates the transition-only
dynamic program in the $(\cost,\nval)$ plane on this corpus: it never increases
$\cost$, and it reduces $\nval$ on more than half the sequences.

\subsection{CPU timing}

\begin{center}
\begin{tabular}{lrrr}
\toprule
Method & Total CPU & Mean per sequence & Median per sequence \\
\midrule
Transition-only DP & 9.291s & 0.300 ms & 0.289 ms \\
Hitting-set NValue & 4.348s & 0.140 ms & 0.086 ms \\
Tonal parsimony & 54.780s & 1.766 ms & 0.562 ms \\
\bottomrule
\end{tabular}
\end{center}

Tonal parsimony is slower than the transition-only DP, but remains practical:
the full corpus of approximately 2.5 million chord positions is processed in
under a minute of process CPU time on the test machine.  The hitting-set method
is faster than the former DP in this corpus because it operates on the small
set of distinct chord domains and the fixed 24-tonality universe.

\section{Chord Substitutions}

\subsection{Musical motivation}

The baseline model treats chord-to-tonality compatibility literally: a chord is
compatible with a tonality when its pitches belong to the corresponding scale.
This pitch-containment relation captures diatonic compatibility but fails on
common jazz and popular chromatic functions.  In such repertoires, a
non-diatonic chord may have a clear functional interpretation inside the
current tonality.  For example, \texttt{Bb7} before \texttt{Amaj7} is heard as
a tritone substitute for \texttt{E7}, the ordinary dominant of A.  If this
interpretation is not
available, an analyzer may be forced to introduce an unrelated local tonality
for \texttt{Bb7}, even though the musical function of the chord is to keep the
progression directed toward A.

We therefore enrich the compatibility relation with a weighted catalogue of
non-diatonic interpretations.  We exclude diatonic same-function substitutions
such as tonic-family exchange: those do not change the set of compatible tonal
centers and are not useful for the present tonality-assignment problem.  The
substitutions modeled here allow a chromatic chord to remain attached to a
plausible local tonic.
Jazz substitution rules have also been approached from a formal algebraic
perspective, notably by \citet{cataldo2018}.

Formally, each chord position can be associated not only with a set of
tonalities, but with a set of interpretations
\[
  (t,r,p),
\]
where $t$ is a candidate tonality, $r$ is an interpretation rule, and $p$ is a
substitution penalty.  The three optimization methods still receive ordinary
candidate domains, obtained by projecting these interpretations onto their
tonalities.  After a path has been selected, the chosen tonality at each
position can be explained a posteriori by the lowest-penalty interpretation
that supports it.  In the experiments below, the substitution penalty is used
for posterior reporting, not as an additional optimization criterion.

Some substitutions have an intrinsically recursive, or fixed-point, character:
after replacing a chord by a substitute, the resulting chord may itself be
interpreted functionally inside a larger tonality.  For example,
\texttt{F\#7} can be understood as the tritone substitute of \texttt{C7}, and
\texttt{C7} can then be understood as $V/IV$ in C major.  Rather than running a
general rule engine to a fixed point during optimization, we compute the
musically relevant finite closure in advance and expose it as a flat
chord-to-tonality compatibility relation.  The optimizer therefore sees only a
standard finite-domain problem, while the posterior explanation records the
compound rule that justified the added candidate.

\subsection{Rule set}

The evaluated rule set is explicit and finite, but covers the main
non-diatonic functional reinterpretations used in the experiments:

For reproducibility, the implementation treats substitution rules as named
rule sets rather than as a single universal catalogue.  The corpus tables below
use the broad \texttt{full} rule set, while source-specific validations can
select smaller or different sets, for example Coker-style rules or the
expert-jazz-progressions preset.

\begin{center}
\scriptsize
\begin{tabular}{lllr}
\toprule
Rule & Pattern & Example & Penalty \\
\midrule
Tritone substitute & $\flat II^7 \to I$ & \texttt{Bb7} $\to$ A & 0.75 \\
Backdoor dominant & $\flat VII^7 \to I$ & \texttt{G7} $\to$ A & 1.00 \\
Secondary dominant & $V^7/x$ in major & \texttt{D7} as $V/V$ in C & 1.00 \\
Tritone secondary dominant & $\mathrm{sub}V^7/x$ in major &
\texttt{F\#7} as $\mathrm{sub}V/IV$ in C & 1.25 \\
Tonic dominant color & $I^7$ as blues tonic & \texttt{C7} in C & 0.50 \\
Blues subdominant & $IV^7$ or $IV^{7sus}$ & \texttt{F7} in C & 1.00 \\
Altered primary dominant & $V^{7alt}$ or $V^{7\flat9} \to I$ &
\texttt{G7b9} in C & 1.00 \\
Suspended tonic dominant & $I^{7sus}$ & \texttt{F7sus4} in F & 0.75 \\
Suspended primary dominant & $V^{7sus} \to I$ & \texttt{Eb7sus4} in Ab & 0.75 \\
Lydian tonic color & $I^{maj7\#11}$ & \texttt{Dmaj7\#11} in D & 0.50 \\
Augmented tonic color & $I^+$ & \texttt{C+} in C & 0.50 \\
Minor tonic color & $i$ or $i^7$ & \texttt{Fm} in F minor & 0.50 \\
Minor supertonic & $ii\emptyset$ in harmonic minor &
\texttt{Gm7b5} in F minor & 0.75 \\
Major supertonic & borrowed $ii\emptyset$ in major &
\texttt{Dm7b5} in C & 0.75 \\
Leading-tone diminished & $vii^\circ \to I$ & \texttt{G\#dim} $\to$ A & 0.75 \\
Borrowed minor iv & $iv^m \to I$ & \texttt{Dm} $\to$ A & 1.00 \\
\bottomrule
\end{tabular}
\end{center}

The first, second, and leading-tone rules add major or harmonic-minor tonic
targets.  Secondary-dominant rules add the enclosing major key: for instance,
\texttt{F\#7} can stay in C major as the tritone substitute of \texttt{C7},
which is $V/IV$ in C.  These secondary rules are examples of the precomputed
closure described above.

The rule weights express how exceptional the interpretation is.  They are
used to explain the resulting analysis and could later be added as a secondary
or tertiary optimization criterion.  The present evaluation keeps the
three objectives unchanged, in order to isolate the effect of enlarging the
candidate domains.

\subsection{Corpus results with substitutions}

On the full corpus, the substitution rules expanded 96 unique chord names,
affecting 913,570 chord positions and adding 3,334,200 candidate tonalities.
The table below recomputes the three methods with the enriched compatibility
relation.

\begin{center}
\begin{tabular}{lrrrr}
\toprule
Method & Mean $\nval$ & Median $\nval$ & Mean $\cost$ & Median $\cost$ \\
\midrule
Transition-only DP & 4.320 & 4.000 & 12.141 & 8.000 \\
Hitting-set NValue & 2.532 & 2.000 & 30.542 & 24.000 \\
Tonal parsimony & 3.206 & 3.000 & 12.141 & 8.000 \\
\bottomrule
\end{tabular}
\end{center}

\begin{figure}[t]
\centering
\includegraphics[width=0.95\linewidth]{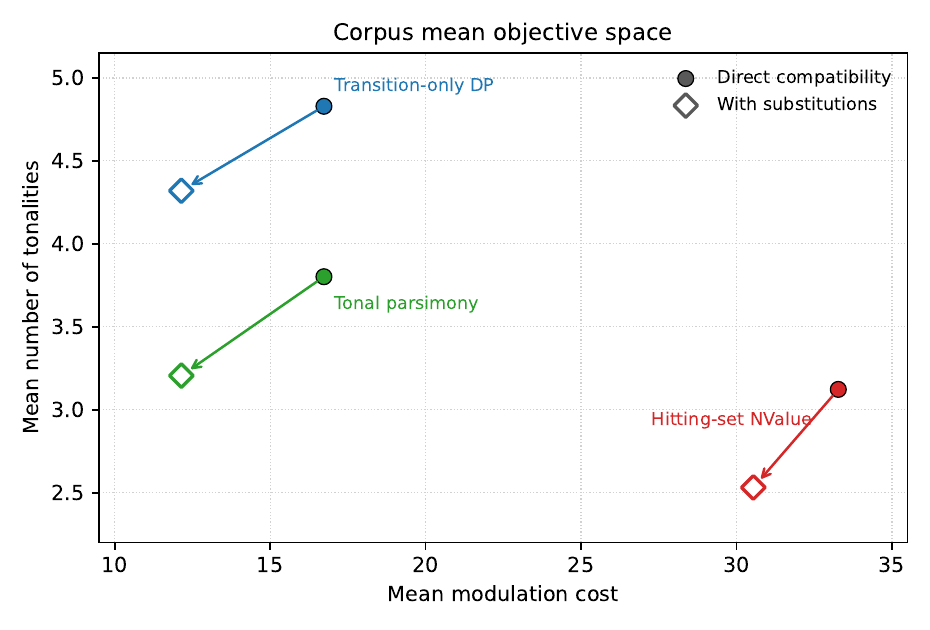}
\caption{Corpus mean objective space for the three methods.  Filled circles
show direct chord-scale compatibility, open diamonds show the same methods
with the substitution closure, and arrows indicate the movement induced by the
substitution layer.  Lower-left is better: lower modulation cost and fewer
distinct tonalities.}
\label{fig:objective-space}
\end{figure}

Compared with the direct-compatibility baseline, substitutions change many
paths and generally reduce the number of tonalities.  For tonal parsimony, the
mean number of tonalities decreases from 3.802 to 3.206, and the mean
transition count decreases from 16.728 to 12.141.

\begin{center}
\scriptsize
\begin{tabular}{lrrrr}
\toprule
Method & Changed paths & Mean $\Delta\nval$ & Mean $\Delta\cost$ &
Selected substitutions \\
\midrule
Transition-only DP & 23,597 / 31,017 & -0.509 & -4.587 & 173,646 \\
Hitting-set NValue & 23,494 / 31,017 & -0.591 & -2.732 & 202,553 \\
Tonal parsimony & 22,147 / 31,017 & -0.596 & -4.587 & 167,004 \\
\bottomrule
\end{tabular}
\end{center}

The selected substitution counts for tonal parsimony are:

\begin{center}
\begin{tabular}{lr}
\toprule
Rule & Selected positions \\
\midrule
Secondary dominant & 64,648 \\
Borrowed minor iv & 41,166 \\
Tonic dominant color & 18,583 \\
Major supertonic & 10,456 \\
Blues subdominant & 6,117 \\
Backdoor dominant & 6,114 \\
Augmented tonic color & 5,585 \\
Minor tonic color & 5,548 \\
Tritone secondary dominant & 4,874 \\
Tritone substitute & 2,524 \\
Leading-tone diminished & 1,389 \\
\bottomrule
\end{tabular}
\end{center}

\section{Professional Chord-Scale Validation}
\label{sec:expert-jazz-validation}

The large corpus experiments above evaluate the objectives themselves.  We
also ran a full external validation against professionally analyzed jazz
standards from the \emph{Jazz Standards Progressions Book}
corpus~\citep{jazzstandardsbook}.  This source provides two annotations for each
chord: a hierarchical Roman-function
analysis and a local chord-scale label.  These are related but not identical
tasks.  The Roman layer may keep a secondary ii--V inside a global key, while
the chord-scale layer tells an improviser which local scale to play.

The validation covers the five C-instrument volumes: 1,555 tune analyses and
52,643 chord positions.  The evaluation compares the local chord-scale label
induced by each selected tonality with the printed chord-scale annotation.
The conversion from tonality path to chord-scale label is deterministic.  For
each chord position, the optimizer first selects a tonality, and the
substitution layer supplies the lowest-penalty interpretation that supports
that tonality.  The chord suffix and scale degree then determine the reported
local scale: major and sixth chords yield Ionian or Lydian colors; minor
seventh chords yield Dorian, Aeolian, or Phrygian colors according to their
degree in the selected major tonality; minor-sixth chords yield melodic minor;
half-diminished chords yield Locrian; diminished chords yield diminished;
plain dominants yield Mixolydian; and explicit dominant colors map to their
corresponding labels, e.g. \texttt{7b9} to Mixolydian flat-nine
flat-thirteen, \texttt{7\#9} to major blues, \texttt{7\#11} to Lydian
dominant, and \texttt{7\#5}/\texttt{7alt} to altered.  Substitution rules
therefore change the selected tonal center, but the final scale label is still
computed from the selected tonality and the observed chord suffix.

We report three increasingly permissive scores:

\begin{itemize}
\item \emph{Strict}: exact agreement after harmless spelling aliases such as
\texttt{Ion} and \texttt{Ionian}.
\item \emph{Compatible}: strict agreement plus conservative, chord-aware
ambiguities.  For example, a major seventh chord may admit Ionian or Lydian
color, a minor seventh chord may admit Dorian, Aeolian, or Phrygian color, and
a plain dominant may admit several dominant-scale colors.  Explicit altered
suffixes are not freely merged at this level.
\item \emph{Broad}: an upper bound that also accepts same-family color
substitutions even when the chord suffix is more explicit, such as altered
dominant variants.  This score diagnoses residual same-family disagreements
but is too permissive to be the main validation criterion.
\end{itemize}

\begin{center}
\scriptsize
\begin{tabular}{lrrrrrr}
\toprule
Method & Strict & Compatible & Broad & Exact comp. tunes & Mean $\nval$ & Mean $\cost$ \\
\midrule
Transition-only DP & 0.798 & 0.943 & 0.995 & 600/1555 & 3.354 & 4.571 \\
Hitting-set NValue & 0.774 & 0.943 & 0.995 & 714/1555 & 2.329 & 10.032 \\
Tonal parsimony & 0.811 & 0.956 & 0.995 & 767/1555 & 2.655 & 4.571 \\
\bottomrule
\end{tabular}
\end{center}

All three methods reach the same broad upper bound, because many of the
remaining disagreements are same-family color choices.  The conservative
compatible score is more informative: tonal parsimony improves from 94.3\% to
95.6\% agreement, and outperforms both other methods on 100 tunes.
The residual hard-disagreement count is small and shared by all methods
(242 events, plus 8 out-of-model labels), suggesting that most remaining
disagreements are either source-label alignment issues or scale categories
outside the 24-tonality model, rather than optimization failures.

Two examples illustrate the difference.  In \emph{The Gift (Recado Bossa
Nova)}, tonal parsimony obtains perfect compatible chord-scale agreement while
also matching the best value of each objective: it uses the two tonalities of
the pure NValue solution but keeps the four transitions of the former DP.  In
\emph{Nuages}, tonal parsimony again reaches perfect compatible agreement,
with the same transition count as the DP and the same number of tonalities as
the pure NValue method.

\begin{center}
\scriptsize
\begin{tabular}{llrrrr}
\toprule
Tune & Method & Compatible & Strict & $\nval$ & $\cost$ \\
\midrule
The Gift & Transition-only DP & 0.909 & 0.879 & 3 & 4 \\
The Gift & Hitting-set NValue & 0.727 & 0.727 & 2 & 14 \\
The Gift & Tonal parsimony & 1.000 & 0.909 & 2 & 4 \\
\midrule
Nuages & Transition-only DP & 0.911 & 0.889 & 5 & 10 \\
Nuages & Hitting-set NValue & 0.911 & 0.822 & 3 & 28 \\
Nuages & Tonal parsimony & 1.000 & 0.844 & 3 & 10 \\
\bottomrule
\end{tabular}
\end{center}

These cases are representative of the reason tonal parsimony matters.  The
former DP often keeps transitions low but carries extra tonalities.  The pure
NValue method compresses tonal vocabulary but may jump frequently.  The
lexicographic method keeps the smooth path while avoiding unnecessary tonal
centers, which in turn reduces broad-only chord-scale color disagreements such
as minor chords being interpreted as melodic-minor tonic color when the printed
annotation expects a contextual Dorian, Aeolian, or Phrygian color.

\section{Coker-Style Functional Validation}
\label{sec:coker-validation}

The corpus experiment measures aggregate behavior, but it does not by itself
say whether the additional compatibility rules correspond to ordinary
jazz-harmonic expectations.  We therefore built a targeted validation suite,
CokerBench, made of 19 short Coker-style functional examples.  The
suite is not used for training or parameter fitting.  It is a regression test
for the chord-to-tonality model and for the three optimization objectives.
Each example specifies an expected number of tonalities, and when the intended
analysis is unambiguous it also specifies an expected sequence of tonic roots.

\subsection{Test families}

The examples cover six kinds of harmonic behavior:

\begin{center}
\scriptsize
\begin{tabular}{llp{0.52\linewidth}}
\toprule
Family & Cases & What is tested \\
\midrule
Basic tonal progressions & 1, 6 &
Major and minor ii--V--I behavior, including altered minor dominants. \\
Single substitutions & 2--5 &
Secondary dominants, tritone substitutes, backdoor dominants, and the
precomputed closure needed for tritone secondary dominants. \\
Combined substitutions & 12 &
Backdoor and tritone substitutions inside one C-centered progression. \\
Jazz forms & 8--11 &
Rhythm-changes A section, rhythm-changes bridge, traditional blues, and a
Bird-blues-like dominant cycle. \\
Dominant and modal colors & 15--17, 19 &
Whole-tone or altered dominant color, diminished-dominant color, sus dominant
regions, and explicit Lydian tonic colors. \\
Objective and ambiguity stress tests & 7, 13, 14, 18 &
Relative-major/minor ambiguity, Coker-inspired chromatic progressions, and
anti-compression cases where the expected analysis should keep several local
tonics. \\
\bottomrule
\end{tabular}
\end{center}

The validation serves a different role from the large corpus evaluation.  In
the corpus, there is no ground-truth analysis, so we compare objectives.  In
CokerBench, the goal is to check whether the analysis remains musically
plausible on named functional situations.  It tests the substitution layer
directly: without substitutions, all three methods solve only 5 of the 18 cases
for which an expected $\nval$ is specified and only 2 of the 15 cases with an
expected root analysis.

\subsection{Results}

With substitution-enriched domains, the three methods solve the same number of
validation cases, but they do not become identical objectives.

\begin{center}
\scriptsize
\begin{tabular}{llrrrr}
\toprule
Compatibility & Method & $\nval$ ok & Root analysis ok & Mean $\nval$ & Mean $\cost$ \\
\midrule
Direct & Transition-only DP & 5/18 & 2/15 & 3.053 & 3.000 \\
Direct & Hitting-set NValue & 5/18 & 2/15 & 2.789 & 3.263 \\
Direct & Tonal parsimony & 5/18 & 2/15 & 2.842 & 3.000 \\
\midrule
Substitutions & Transition-only DP & 16/18 & 13/15 & 1.526 & 0.579 \\
Substitutions & Hitting-set NValue & 16/18 & 13/15 & 1.421 & 0.895 \\
Substitutions & Tonal parsimony & 16/18 & 13/15 & 1.421 & 0.579 \\
\bottomrule
\end{tabular}
\end{center}

Most CokerBench examples become easy once the correct functional
interpretations are available.  The three methods choose identical paths on 16
of the 19 cases and have identical $(\cost,\nval)$ values on 17 cases.  This
is not a weakness of the validation suite: it shows that many apparent
optimization disagreements in literal chord-scale analysis were actually
missing-domain problems.  Once a \texttt{Db7} can remain in C as a tritone
substitute, or \texttt{F7} can remain in C as a blues subdominant, the intended
analysis often has both low $\cost$ and low $\nval$.

There remains one clear three-way objective example, the Two-Bop-like
progression
\[
\texttt{Abmaj7 | Am7 | Bbm7 | C7 | Fm7 | Bb7 | Ebm7 | Ab7 | Dbmaj7 | Gb7 | Em7 | A7}.
\]
The results are:

\begin{center}
\scriptsize
\begin{tabular}{lrrp{0.58\linewidth}}
\toprule
Method & $\cost$ & $\nval$ & Analysis path \\
\midrule
Transition-only DP & 3 & 4 &
D\#maj; A h.min.; C\#maj x8; Cmaj x2 \\
Hitting-set NValue & 9 & 3 &
D\#maj; Cmaj; C\#maj; D\#maj x3; C\#maj; D\#maj; C\#maj; D\#maj; Cmaj; D\#maj \\
Tonal parsimony & 3 & 3 &
D\#maj; Cmaj; C\#maj x7; Cmaj x3 \\
\bottomrule
\end{tabular}
\end{center}

This is the intended behavior of the proposed method.  The former DP keeps the
lowest transition count but uses four tonalities.  The hitting-set method uses
three tonalities but introduces nine transitions.  Tonal parsimony obtains the
same three-tonality vocabulary as the hitting-set method while preserving the
transition count of the dynamic program.

\subsection{Remaining misses}

The two remaining tonal-parsimony failures are musically informative.

The first is the relative-major excursion
\[
\texttt{Fm | Bbm7 | Eb7 | Abmaj7 | Gm7b5 | C7alt | Fm}.
\]
The expected analysis is a two-region reading: F minor around the opening and
closing ii--V--i material, with a brief Ab-major region at \texttt{Abmaj7}.
The substitution-enriched tonal-parsimony analysis instead assigns the whole
sequence to \texttt{G\#maj}, enharmonically Ab major:

\begin{center}
\scriptsize
\begin{tabular}{lrrp{0.50\linewidth}}
\toprule
Reading & $\cost$ & $\nval$ & Path \\
\midrule
Expected & 2 & 2 & F minor x3; Ab major; F minor x3 \\
Tonal parsimony & 0 & 1 & G\#maj x7 \\
\bottomrule
\end{tabular}
\end{center}

This is a limitation of the objective rather than a missing substitution.  The
relative major covers much of the same pitch material, and the current
criterion prefers the zero-transition one-tonality explanation.  A more refined
model would need local functional evidence or a tonic-minor preference strong
enough to distinguish ``contained in Ab major'' from ``functioning in F
minor''.

The second miss is an anti-compression example:
\[
\texttt{Cmaj7 | Ebmaj7 | Abmaj7 | Bmaj7 | Gmaj7}.
\]
Here the intended analysis uses five local major tonics.  The current direct
compatibility relation allows a major-seventh chord to be interpreted not only
as tonic, but also as a diatonic chord in another major scale.  Tonal parsimony
therefore compresses the example to three tonalities:

\begin{center}
\scriptsize
\begin{tabular}{lrrp{0.50\linewidth}}
\toprule
Reading & $\cost$ & $\nval$ & Path \\
\midrule
Expected & 4 & 5 & Cmaj; Ebmaj; Abmaj; Bmaj; Gmaj \\
Tonal parsimony & 3 & 3 & Gmaj; D\#maj x2; F\#maj; Gmaj \\
\bottomrule
\end{tabular}
\end{center}

The related Lydian-color case with explicit \texttt{maj7\#11} labels is solved
correctly by the new rule layer.  This plain-\texttt{maj7} failure shows that
some anti-compression examples require weighted local evidence in addition to
more substitutions: a bare \texttt{Cmaj7} is pitch-compatible with G major, but
in many jazz contexts it may still be intended as a C tonic sonority.

\section{Illustrative Cases}
\label{sec:illustrative-cases}

The examples below show why both ingredients matter: the lexicographic
objective and the substitution layer.  Paths are shown in run-length form: for
example, \texttt{Cmaj x4} means that C major is
assigned to four consecutive chord positions.  We abbreviate harmonic minor as
\texttt{h.min.}.

\subsection{A compact six-way example}

The C-centered progression
\[
\texttt{Cmaj7 | F\#7 | F | Bb7 | Em | Am | Dm | Db7}
\]
isolates the substitution layer: \texttt{Bb7} can be heard as a backdoor
dominant into C, and \texttt{Db7} as a tritone substitute resolving to C.  The
less obvious chord is \texttt{F\#7}: it can be heard as the tritone
substitute of \texttt{C7}, which is itself a secondary dominant of IV in C.
In the present implementation, substitutions reduce tonal parsimony from
$(\cost,\nval)=(5,4)$ to $(0,1)$ on this progression.  The three methods tie
after substitutions, however, so the next excerpt is better for comparing the
objectives themselves.

The following short excerpt is enough to show the three objective behaviors
and the effect of substitution-enriched domains:
\[
\texttt{Dm | Em | F\#m | Bm | Gm6 | D | Bm | C\#m7}.
\]
The column ``Subs.'' reports how many selected chord positions are explained
by a substitution rule in the posterior interpretation.

\begin{center}
\scriptsize
\begin{tabular}{llrrrp{0.42\linewidth}}
\toprule
Compatibility & Method & $\cost$ & $\nval$ & Subs. & Analysis path \\
\midrule
Direct & Transition-only DP & 3 & 4 & 0 &
Cmaj x2; Dmaj x2; Fmaj; Amaj x3 \\
Direct & Hitting-set NValue & 5 & 3 & 0 &
Fmaj; G\# h.min.; F\# h.min. x2; Fmaj; F\# h.min. x2; G\# h.min. \\
Direct & Tonal parsimony & 3 & 3 & 0 &
Cmaj x2; Amaj; B h.min. x2; Amaj x3 \\
Substitutions & Transition-only DP & 2 & 3 & 1 &
Cmaj x2; Dmaj x5; Emaj \\
Substitutions & Hitting-set NValue & 4 & 2 & 2 &
Amaj; Dmaj; Amaj x2; Dmaj; Amaj x3 \\
Substitutions & Tonal parsimony & 2 & 2 & 2 &
Amaj; Dmaj x4; Amaj x3 \\
\bottomrule
\end{tabular}
\end{center}

The final row is the desired behavior.  Direct tonal parsimony keeps the
minimum transition count of the former dynamic program but uses one fewer
tonality.  Pure NValue minimization also uses fewer tonalities, but it pays two
extra transitions.  With substitutions, the lexicographic method keeps the best
transition count among the substitution-aware analyses and reaches the
two-tonality vocabulary of the pure NValue solution.  In this local
$(\cost,\nval)$ comparison, it sits at the useful corner: the same modulation
cost as transition-only DP, but the same tonal-vocabulary size as NValue.

The posterior explanations are simple: \texttt{Dm} is heard as borrowed minor
iv in A major, and \texttt{Gm6} as borrowed minor iv in D major.  These two
interpretations turn the path into a readable alternation between A major and
D major, instead of forcing a local F-major detour.

\subsection{A transition-only tie resolved by tonal parsimony}

One corpus slice exhibits the typical benefit of the lexicographic method:
\[
\begin{array}{l}
\texttt{C | Cm7 | Fm7 | Gmaj7 | Cm | Bmaj7 | C | Cm7 | Fm7 | Gmaj7 |} \\
\texttt{Cm | Bmaj7 | C | G | F | C | F | C | G | C | Gmaj7 | C | Cm7 | Fm7}
\end{array}
\]

\begin{center}
\scriptsize
\begin{tabular}{lrrp{0.58\linewidth}}
\toprule
Method & $\cost$ & $\nval$ & Analysis path \\
\midrule
Transition-only DP & 12 & 5 &
Cmaj; D\#maj x2; Dmaj; D\#maj; F\#maj; Cmaj; D\#maj x2; Dmaj;
D\#maj; F\#maj; Cmaj x8; Gmaj x2; D\#maj x2 \\
Hitting-set NValue & 13 & 4 &
Cmaj; D\#maj x2; Gmaj; D\#maj; F\#maj; Cmaj; D\#maj x2; Gmaj;
D\#maj; F\#maj; Cmaj x8; Gmaj; Cmaj; D\#maj x2 \\
Tonal parsimony & 12 & 4 &
Gmaj; D\#maj x2; Gmaj; D\#maj; F\#maj; Gmaj; D\#maj x2; Gmaj;
D\#maj; F\#maj; Cmaj x5; Gmaj x5; D\#maj x2 \\
\bottomrule
\end{tabular}
\end{center}

The transition-only DP and tonal-parsimony analyses both have $\cost=12$, but
the transition-only path uses five tonalities whereas tonal parsimony uses
four.  The important point is visible in the path: the transition-only optimum
contains an avoidable D major region around the \texttt{Gmaj7} chords, whereas
tonal parsimony can choose G major there without increasing the modulation
count.  The hitting-set method also uses four tonalities, but pays one
additional modulation.

Additional corpus examples show the same trade-off at larger scale: pure
NValue can compress the tonal vocabulary further, but it often does so by
adding transitions.  We omit those longer paths here because the aggregate
tables and the professional validation already quantify the effect.

\section{Discussion}

The results show that the transition-only objective is under-specified: many
minimum-modulation paths exist, and they can differ in the number of tonalities
they introduce.  Adding a secondary NValue criterion resolves this ambiguity
toward smaller tonal vocabularies.

This is the musical reason for minimizing $\nval$.  The value $\nval$ measures
the size of the tonal vocabulary required by an analysis.  Two analyses may
have identical modulation cost, but one may require a larger set of tonal
centers to explain the same progression.  Minimizing $\nval$ favors analyses
that reuse previously established tonal material and therefore produce more
compact scale maps for improvisation, arrangement, and composition.

Pure NValue minimization is informative as a boundary case.
It shows how far the tonal vocabulary can be compressed, while demonstrating
why transition cost cannot be ignored.  The method often produces analyses
with many abrupt changes.

The proposed tonal-parsimony objective combines the established
modulation-minimization principle with a global notion of parsimony.

The substitution experiment shows that the chord-to-tonality relation itself is
also musically important.  A purely literal relation can force the optimizer to
explain conventional chromatic chords by introducing local tonalities that are
pitch-compatible but functionally implausible.  Weighted substitution
interpretations address this without changing the three optimization methods:
they enrich the candidate domains and provide a posterior explanation of which
chromatic events were used in the selected path.

\section{Generating Under a Fixed Tonality Budget}
\label{sec:generation-budget}

The purpose of this section is to demonstrate that the optimized value $\nval$
behaves as an interpretable measure of harmonic complexity on generated
material.  It does not solve the harder problem of generating directly under a
prescribed optimized $\nval$.

The analysis problem studied here is already combinatorial, but it remains
tractable because the tonality universe is fixed and small.  A more general
generative problem is harder: generate or transform chord sequences such that
the minimized analysis uses a prescribed number $\nval$ of tonalities.  This is
not a local constraint on chords.  The value of $\nval$ is a global property of
the optimal analysis after minimization.  Changing one chord can alter the
candidate domains, the minimum modulation structure, and the set of tonalities
selected by the optimizer.

The motivation is practical.  In an interactive generation setting, a user may
want a progression that is ``simple'', ``moderately rich'', or ``more
adventurous'' without having to specify explicit modulations, chord functions,
or substitution rules.  Although harmonic analysis itself requires musical
knowledge, the optimized value $\nval$ induces a compact and audible notion of
complexity: $\nval=1$ means that the whole sequence can be heard through one
tonal center, while larger values require a larger tonal vocabulary to explain
the same surface chords.  This makes $\nval$ an operational control parameter for
adaptive generation, where the system can match the amount of harmonic
movement to the user's intention instead of sampling progressions at random.

We evaluate this generate-then-cluster protocol by training a variable-order
chord generator on the first 1,500 extracted mDecks chord sequences and
generating 1,000 eight-chord continuations with two local constraints: the first
and last generated chords had to be unslashed major triads or major-seventh
chords.  Generation used a max-order 4 variable-order Markov model with exact
positional constraints, using sparse context-state belief propagation for
regular-constrained generation~\citep{pachet2026regular}; the resulting
sequences were then analyzed with the substitution-enriched tonal parsimony
model.  In this baseline, $\nval$ is assigned after analysis rather than
enforced during generation.

\begin{center}
\scriptsize
\begin{tabular}{rrrr}
\toprule
$\nval$ & Samples & Share & Unique sequences \\
\midrule
1 & 226 & 22.6\% & 213 \\
2 & 542 & 54.2\% & 525 \\
3 & 214 & 21.4\% & 206 \\
4 & 17 & 1.7\% & 17 \\
5 & 1 & 0.1\% & 1 \\
\bottomrule
\end{tabular}
\end{center}

All 1,000 generated sequences were analyzable; 962 were distinct.  The mean
post-analysis value was $\nval=2.025$, with mean transition cost
$\cost=1.276$.  The distribution is not concentrated on a single value:
$\nval=1$, $\nval=2$, and $\nval=3$ account respectively for 22.6\%, 54.2\%,
and 21.4\% of the generated continuations.  This suggests that tonal vocabulary
size behaves as a genuine structural property of generated harmony.  The
following examples are selected from the same run, choosing low-substitution,
low-slash representatives within each cluster when possible.

\begin{center}
\scriptsize
\begin{tabular}{rrp{0.47\linewidth}p{0.25\linewidth}}
\toprule
$\nval$ & $\cost$ & Chord sequence & Tonal-parsimony path \\
\midrule
1 & 0 &
Cmaj7; Am7; Dm7; G7; Cmaj7; Am7; Dm7; Fmaj7 &
C major x8 \\
2 & 1 &
G; D; G; A; G; Am7; D7; Gmaj7 &
D major x4; G major x4 \\
3 & 2 &
Cmaj7; Am7; Dm7; G7; Cmaj7; Gm7; Fmaj7; Ebmaj7 &
C major x5; F major x2; A\# major \\
4 & 3 &
Cmaj7; Fm7; Bb7; Ebmaj7; Abm7; Db7; Gbmaj7; Amaj7 &
C major; D\# major x3; F\# major x3; A major \\
5 & 4 &
Bbmaj7; Gbmaj7; Db7\#11; Gm7; Gbm7; Fm7; Ebmaj7; Emaj7 &
F major; C\# major; D major x3; D\# major x2; E major \\
\bottomrule
\end{tabular}
\end{center}

Thus, generation with a target post-analysis $\nval$ combines sequence
generation with an embedded global optimization problem.  It is closer to
constraint-based generation with an oracle objective than to ordinary local
Markov or dynamic-programming generation.  It defines a separate bilevel
constrained-generation problem.

\section{Conclusion}

We introduced and compared three methods for assigning tonalities to chord
sequences.  The former transition-only dynamic program is efficient and
minimizes modulations, but it does not control the number of tonalities used.
The pure hitting-set method minimizes tonal vocabulary but can produce many
modulations.  The proposed tonal-parsimony method minimizes modulations first
and tonal vocabulary second.  This joint objective is much harder
combinatorially than the classical transition-only dynamic program: the
optimizer must choose a smooth path while also minimizing the set of distinct
tonalities appearing anywhere in that path.  The practical contribution is that
the hard part is parameterized by the number of possible tonalities, which is
fixed and small in major/minor tonal music.  We also introduced a substitution
layer for non-diatonic jazz interpretations, covering tritone substitutes,
backdoor and secondary dominants, tonic colors, borrowed chords, and related
functional reinterpretations as weighted candidate-domain enrichments.

On the LMD Chords corpus of 31,032 chord sequences~\citep{holloway2025lmdchords},
tonal parsimony preserves the transition optimum of the former dynamic program
in all analyzed cases and reduces the number of tonalities in 55.8\% of them.
Although the method is exponential in the number of tonalities in the general
case, the fixed 24-tonality setting of major/minor tonal music makes it
efficient in practice.
With substitution-enriched domains, tonal parsimony changes 71.4\% of the
analyzed paths, lowers the mean number of tonalities from 3.802 to 3.206, and
lowers the mean number of modulations from 16.728 to 12.141.  Musically, it
produces more realistic analyses that reconcile the various optimization
criteria in a single global optimization procedure.  The
substitution rules are best understood as a finite closure of repeated
functional interpretations, flattened into ordinary chord-to-tonality
candidates before optimization.  Thus substitutions do not require a separate
post-hoc repair procedure: they are part of the same finite-domain model and
are optimized at scale.

The generation experiment further suggests that $\nval$ behaves as a
nontrivial structural descriptor of harmonic complexity: it is not concentrated
on a single value, but separates generated continuations into distinct
tonal-vocabulary regimes.  This makes $\nval$ useful as a post-analysis control
variable, while enforcing a target $\nval$ during generation remains a separate
bilevel constrained-generation problem.

The external validation strengthens this conclusion on a corpus of
professional chord-scale analyses.  On 1,555 tunes, tonal parsimony has the
best strict chord-scale agreement (81.1\%) and the best conservative
compatible agreement (95.6\%), compared with 94.3\% compatible agreement for
both baselines.  It outperforms both baselines on 100 individual tunes; in
representative cases such as \emph{The Gift} and \emph{Nuages}, it
combines the low transition count of the DP with the small tonal vocabulary of
the NValue solution.  Beyond improving its own objective on corpus statistics,
the method better predicts professionally printed local chord-scale annotations
under a conservative, musically meaningful success criterion.  The resulting
picture is that minimizing both transition cost and tonal vocabulary is
combinatorially difficult in general, but tractable for the tonal universe
relevant to jazz harmony, and extends to professional-level substitutional
analyses over large corpora.

\bibliographystyle{tMAM}
\bibliography{TonalParsimonyAuthorBibTexDatabase}

@book{rossi2006,
  editor = {Rossi, Francesca and van Beek, Peter and Walsh, Toby},
  title = {Handbook of Constraint Programming},
  publisher = {Elsevier},
  year = {2006}
}

@article{temperley1999,
  author = {Temperley, David},
  title = {What's Key for Key? The {Krumhansl--Schmuckler} Key-Finding Algorithm Reconsidered},
  journal = {Music Perception},
  volume = {17},
  number = {1},
  pages = {65--100},
  year = {1999},
  doi = {10.2307/40285812},
  url = {https://doi.org/10.2307/40285812}
}

@article{choi2011,
  author = {Choi, Andrew},
  title = {Jazz Harmonic Analysis as Optimal Tonality Segmentation},
  journal = {Computer Music Journal},
  volume = {35},
  number = {2},
  pages = {49--66},
  year = {2011},
  doi = {10.1162/COMJ_a_00056},
  url = {https://doi.org/10.1162/COMJ_a_00056}
}

@article{raphaelstoddard2004,
  author = {Raphael, Christopher and Stoddard, Josh},
  title = {Functional Harmonic Analysis Using Probabilistic Models},
  journal = {Computer Music Journal},
  volume = {28},
  number = {3},
  pages = {45--52},
  year = {2004},
  doi = {10.1162/0148926041790676},
  url = {https://doi.org/10.1162/0148926041790676}
}

@misc{bouquillard2024,
  author = {Bouquillard, Augustin and Jacquemard, Florent},
  title = {Engraving Oriented Joint Estimation of Pitch Spelling and Local and Global Keys},
  howpublished = {arXiv preprint arXiv:2402.10247},
  year = {2024},
  url = {https://arxiv.org/abs/2402.10247}
}

@book{chew2014,
  author = {Chew, Elaine},
  title = {Mathematical and Computational Modeling of Tonality: Theory and Applications},
  publisher = {Springer},
  year = {2014},
  doi = {10.1007/978-1-4614-9475-1},
  url = {https://doi.org/10.1007/978-1-4614-9475-1}
}

@article{tsushima2018,
  author = {Tsushima, Hiroaki and Nakamura, Eita and Itoyama, Katsutoshi and Yoshii, Kazuyoshi},
  title = {Generative Statistical Models with Self-Emergent Grammar of Chord Sequences},
  journal = {Journal of New Music Research},
  volume = {47},
  number = {3},
  pages = {226--248},
  year = {2018},
  doi = {10.1080/09298215.2018.1447584},
  url = {https://doi.org/10.1080/09298215.2018.1447584}
}

@article{steedman1984,
  author = {Steedman, Mark J.},
  title = {A Generative Grammar for Jazz Chord Sequences},
  journal = {Music Perception},
  volume = {2},
  number = {1},
  pages = {52--77},
  year = {1984},
  doi = {10.2307/40285282},
  url = {https://doi.org/10.2307/40285282}
}

@article{cataldo2018,
  author = {Cataldo, Carmine},
  title = {Towards a Music Algebra: Fundamental Harmonic Substitutions in Jazz},
  journal = {International Journal of Advanced Engineering Research and Science},
  volume = {5},
  number = {1},
  pages = {52--57},
  year = {2018},
  doi = {10.22161/ijaers.5.1.9},
  url = {https://doi.org/10.22161/ijaers.5.1.9}
}

@misc{jazzstandardsbook,
  author = {{Jazz Standards Progressions Book}},
  title = {The Jazz Standards Progressions Book: C Instruments},
  howpublished = {Volumes 1--5},
  year = {n.d.},
  note = {Jazz standards chord-progressions and chord-scale annotations}
}

@misc{pachet2026regular,
  author = {Pachet, Fran{\c{c}}ois},
  title = {Exact Regular-Constrained Variable-Order Markov Generation via Sparse Context-State Belief Propagation},
  howpublished = {arXiv preprint arXiv:2605.07839},
  year = {2026},
  url = {https://arxiv.org/abs/2605.07839}
}

@misc{holloway2025lmdchords,
  author = {Holloway, Oliver},
  title = {{lmd\_chords} (Revision 4d6815c)},
  howpublished = {Hugging Face dataset},
  year = {2025},
  doi = {10.57967/hf/4219},
  url = {https://huggingface.co/datasets/ohollo/lmd_chords},
  note = {Hugging Face}
}

@phdthesis{raffel2016,
  author = {Raffel, Colin},
  title = {Learning-Based Methods for Comparing Sequences, with Applications to Audio-to-{MIDI} Alignment and Matching},
  school = {Columbia University},
  year = {2016}
}

@inproceedings{mauch2010,
  author = {Mauch, Matthias and Dixon, Simon},
  title = {Approximate Note Transcription for the Improved Identification of Difficult Chords},
  booktitle = {Proceedings of the 11th International Society for Music Information Retrieval Conference},
  pages = {135--140},
  publisher = {International Society for Music Information Retrieval},
  year = {2010}
}

@incollection{karp1972,
  author = {Karp, Richard M.},
  title = {Reducibility among Combinatorial Problems},
  booktitle = {Complexity of Computer Computations},
  editor = {Miller, Raymond E. and Thatcher, James W.},
  pages = {85--103},
  publisher = {Plenum Press},
  year = {1972},
  doi = {10.1007/978-1-4684-2001-2_9},
  url = {https://doi.org/10.1007/978-1-4684-2001-2_9}
}

@inproceedings{beldiceanu2000,
  author = {Beldiceanu, Nicolas},
  title = {Global Constraints as Graph Properties on a Structured Network of Elementary Constraints of the Same Type},
  booktitle = {Principles and Practice of Constraint Programming -- {CP} 2000, Lecture Notes in Computer Science 1894},
  pages = {52--66},
  publisher = {Springer},
  year = {2000},
  doi = {10.1007/3-540-45349-0},
  url = {https://doi.org/10.1007/3-540-45349-0}
}

@inproceedings{pesant2004,
  author = {Pesant, Gilles},
  title = {A Regular Language Membership Constraint for Finite Sequences of Variables},
  booktitle = {Principles and Practice of Constraint Programming -- {CP} 2004, Lecture Notes in Computer Science 3258},
  pages = {482--495},
  publisher = {Springer},
  year = {2004},
  doi = {10.1007/978-3-540-30201-8_36},
  url = {https://doi.org/10.1007/978-3-540-30201-8_36}
}

@article{demassey2006,
  author = {Demassey, Sophie and Pesant, Gilles and Rousseau, Louis-Martin},
  title = {A Cost-Regular Based Hybrid Column Generation Approach},
  journal = {Constraints},
  volume = {11},
  number = {4},
  pages = {315--333},
  year = {2006},
  doi = {10.1007/s10601-006-9003-7},
  url = {https://doi.org/10.1007/s10601-006-9003-7}
}

@inproceedings{bessiere2006,
  author = {Bessiere, Christian and Hebrard, Emmanuel and Hnich, Brahim and Kiziltan, Zeynep and Walsh, Toby},
  title = {Filtering Algorithms for the \textsc{NValue} Constraint},
  booktitle = {Integration of {AI} and {OR} Techniques in Constraint Programming for Combinatorial Optimization Problems, Lecture Notes in Computer Science 3524},
  pages = {79--93},
  publisher = {Springer},
  year = {2005},
  doi = {10.1007/11493853_8},
  url = {https://doi.org/10.1007/11493853_8}
}

\addcontentsline{toc}{section}{References}

\end{document}